\documentclass[11pt]{article}

\newcommand{\blind}{1}

\usepackage[inline]{enumitem}
\usepackage{float}
\usepackage{amsmath}
\usepackage{graphicx}
\usepackage[round]{natbib}
\usepackage{float}
\usepackage{geometry}
\usepackage{tikz}
\usepackage[english]{babel}
\usepackage{longtable}
\usepackage{color}
\usepackage{amssymb, amsmath, amsthm}
\usepackage{multirow}
\usepackage[titletoc,title]{appendix}
\usepackage{authblk}
\usepackage{setspace}
\usepackage{dsfont}
\RequirePackage[OT1]{fontenc}
\usepackage{subcaption}

\RequirePackage[colorlinks=true,citecolor=blue,urlcolor=blue,linkcolor=black]{hyperref}
\newtheorem{remark}{Remark}
\newtheorem{theorem}{Theorem}
\newtheoremstyle{indented}{5pt}{5pt}{\addtolength{\leftskip}{2.5em}}{}{\bfseries}{.}{.7em}{}

{
  \theoremstyle{definition}
  \newtheorem{assumption}{}
}
{
  \theoremstyle{definition}
  \newtheorem{assumptioniden}{}
}

\newtheorem{lemma}{Lemma}

\DeclareMathOperator{\expit}{expit}
\DeclareMathOperator{\logit}{logit}
\DeclareMathOperator{\var}{Var}
\DeclareMathOperator{\dr}{dr}
\DeclareMathOperator{\tmlee}{tmle}
\DeclareMathOperator{\aipww}{aipw}
\DeclareMathOperator{\opt}{opt}
\DeclareMathOperator{\dtmlee}{dtmle}
\DeclareMathOperator{\daipww}{daipw}

\newcommand{\Pn}{\mathbb{P}_{n}}
\newcommand{\daipw}{\hat\theta_{\daipww}}
\newcommand{\hopt}{\hat h_{\opt}}
\newcommand{\aipw}{\hat\theta_{\aipww}}
\newcommand{\dtmle}{\hat\theta_{\dtmlee}}
\newcommand{\tmle}{\hat\theta_{\tmlee}}

\newcommand{\thetac}{\theta_{\mbox{\footnotesize causal}}}

\renewenvironment{proof}{{\it Proof }}{\qed \\}

\title{Doubly Robust Inference for Targeted Minimum Loss Based
  Estimation in Randomized Trials with Missing Outcome Data}

\if1\blind
\author[1]{Iv\'an  D\'iaz \thanks{corresponding author: ild2005@med.cornell.edu}}
\author[2]{Mark J. van der Laan}
\affil[1]{\small Division of Biostatistics, Weill Cornell Medicine.}
\affil[2]{\small Division of Biostatistics, University of California
  at Berkeley.}
\fi

\if0\blind
\author[1]{\vspace{-1.5cm} }
\fi

\begin{document}\maketitle

\begin{abstract}
  Missing outcome data is one of the principal threats to the validity
  of treatment effect estimates from randomized trials. The outcome
  distributions of participants with missing and observed data are
  often different, which increases the risk of bias. Causal inference
  methods may aid in reducing the bias and improving efficiency by
  incorporating baseline variables into the analysis. In particular,
  doubly robust estimators incorporate estimates of two nuisance
  parameters: the outcome regression and the missingness mechanism
  (i.e., the probability of missingness conditional on treatment
  assignment and baseline variables), to adjust for differences in the
  observed and unobserved groups that can be explained by observed
  covariates. To obtain consistent estimators of the treatment effect,
  one of these two nuisance parameters mechanism must be consistently
  estimated. Such nuisance parameters are traditionally estimated
  using parametric models, which generally preclude consistent
  estimation, particularly in moderate to high dimensions. Recent
  research on missing data has focused on data-adaptive estimation of
  the nuisance parameters in order to achieve consistency, but the
  large sample properties of such estimators are poorly understood. In
  this article we discuss a doubly robust estimator that is consistent
  and asymptotically normal (CAN) under data-adaptive consistent
  estimation of the outcome regression \textit{or} the missingness
  mechanism.  We provide a formula for an asymptotically valid
  confidence interval under minimal assumptions.  We show that our
  proposed estimator has smaller finite-sample bias compared to
  standard doubly robust estimators. We present a simulation study
  demonstrating the enhanced performance of our estimators in terms of
  bias, efficiency, and coverage of the confidence intervals. We
  present the results of an illustrative example: a randomized,
  double-blind phase II/III trial of antiretroviral therapy in
  HIV-infected persons, and provide R code implementing our proposed
  estimators.
\end{abstract}

% \doublespacing
\section{Introduction}

Missing data are a frequent problem in randomized trials. If the
reasons for outcome missingness and the outcome itself are correlated,
unadjusted estimators of the treatment effect are biased, thus
invalidating the conclusions of the trial. Most methods to mitigate
the bias rely on baseline variables to control for the possible common
causes of missingness and the outcome, through estimation of certain
``nuisance'' parameters, i.e., parameters that are not of interest in
themselves, but that are required to estimate the treatment effect. In
addition to aiding in correcting bias, methods that use covariate
adjustment often provide more precise estimates \cite[see,
e.g.,][]{koch1998issues,Bang05,Zhang2008,moore2009covariate,Colantuoni2015, Diaz2016}. In this
article we focus on doubly robust estimators.  Doubly robust
estimation of treatment effects in randomized trials requires
estimation of two possibly high-dimensional nuisance parameters: the
outcome expectation within treatment arm conditional on baseline
variables (henceforth referred to as outcome regression), and the
probability of missingness conditional on baseline variables
(henceforth referred to as missingness mechanism).

The large sample properties of doubly robust estimators hinges upon
large sample properties of the estimators of the nuisance
parameters. In particular:
\begin{enumerate*}[label=(\alph*)]
\item doubly robust estimators remain consistent if at least one of
  the nuisance parameters is estimated consistently, and\label{prop:a}
\item the asymptotic distribution of the effect estimator depends on
  empirical process conditions on the estimators of the nuisance
  parameters.\label{prop:b}
\end{enumerate*}
% These two properties of doubly robust estimators have important
% consequences. % Doubly robust estimators based on nuisance parametric
% models are asymptotically normal but they may have considerable bias,
% even at very large sample sizes.
When parametric models are adopted to estimate the nuisance
parameters, a straightforward application of the delta method yields
the convergence of the doubly robust estimator to a normal random
variable at $n^{1/2}$-rate. The nonparametric bootstrap or an
influence function based approach yields consistent estimates of the
asymptotic variance and confidence intervals. However, the assumptions
encoded in parametric models are rarely justified by scientific
knowledge. This implies that parametric models are frequently
misspecified, which yields an inconsistent effect estimator. In other
words, a doubly robust estimators relying on nuisance parametric
models makes no use of the double robustness property \ref{prop:a}: it
is always inconsistent.

Data-adaptive alternatives to alleviate this shortcoming have been
developed over the last decades in the statistics and machine learning
literature. These data-adaptive methods offer an opportunity to employ
flexible estimators that are more likely to achieve
consistency. Methods such as those based on regression trees,
regularization, boosting, neural networks, support vector machines,
adaptive splines, etc., and ensembles of them offer flexibility in the
specification of interactions, non-linear, and higher-order terms, a
flexibility that is not available for parametric models. However, the
large sample analysis of treatment effects estimates based on
machine learning requires hard-to-verify assumptions, and often yield
estimators which are not $n^{1/2}$-consistent, and for which no
statistical inference (i.e., p-values and confidence intervals) is
available. Nonetheless, data-adaptive estimation has been widely used
in estimation of causal effects from observational data \citep[a few
examples include][]{vanderLaan&Petersen&Joffe05,
  Wang&Bembom&vanderLaan06,ridgeway2007, Bembometal08a,
  lee2010improving,neugebauer2016case}. Indeed, the statistics field
of \textit{targeted learning} \citep[see e.g.,][]{vanderLaan&Rubin06,
  vanderLaanRose11, van2014entering} is concerned with the development
of optimal ($n^{1/2}$-consistent, asymptotically normal, efficient)
estimators of smooth low-dimensional parameters through the use
state-of-the art machine learning.

We develop estimators for analyzing data from randomized trials with
missing outcomes, when the missingness probabilities and the outcome
regression are estimated with data-adaptive methods. % We provide a
We propose two estimators: an augmented inverse probability weighted
estimator (AIPW), and a targeted minimum loss based estimator
(TMLE). Our methods are inspired by recent work by
\cite{van2014targeted, benkeser2016doubly}, who developed an estimator
of the mean of an outcome from incomplete data when data-adaptive
estimators are used for the missingness mechanism. In addition to
extending their methodology to our problem, our main contribution is
to simplify the assumptions of their theorems to two conditions:
consistent estimation of at least one of the nuisance parameters, and
a condition restricting the class of estimators of the nuisance
parameters to Donsker classes (those for which a uniform central limit
theorem applies). Though the Donsker condition may be removed through
the use of a cross-validated version of our TMLE, the results are
straightforward extensions of the work of \cite{zheng2011cross}, and
we do not pursue such results here. We show that the doubly robust
asymptotic distribution of these novel estimators requires a slightly
stronger version of the standard double robustness in which the
nuisance parameters converge to their (possibly misspecified) limits
at $n^{1/4}$-rate, with at least one of them converging to the correct
limit. Specifically, we show that the TMLE is CAN under these
empirical process conditions, and provide its influence function. This
allows the construction of Wald-type confidence intervals under the
assumption that at least one of the nuisance parameters is
consistently estimated, though it is not necessary to know which
one. We also make connections between the proposed estimators and
standard $M$-estimation theory, by noting that our estimators
\citep[and those of][]{van2014targeted, benkeser2016doubly} amount to
controlling the behavior of the ``drift'' term resulting from the
analysis of the estimator's empirical process. Thus, our methods and
theory may be used to improve the performance of other $M$-estimators
in causal inference and missing data problems. The need to control the
behavior of such terms has been previously recognized in the
semiparametric estimation literature, for example in Theorem 5.31 of
\cite{vanderVaart98} \cite[see also Section 6.6
of][]{bolthausen2002lectures}.

In related work, \cite{vermeulen2015bias, vermeulen2016data} recently
proposed estimators that also target minimization of the drift
term. However, their methods are not suitable for our application
because they rely on parametric working models for the missingness
mechanism. Since we do not know the functional form of the missingness
mechanism, we must resort to data-adaptive methods to estimate this
probability.

The paper is organized as follows. In Section~\ref{sec:applica} we
discuss our illustrative application and define the statistical
estimation
problem. % In Section~\ref{sec:litreview} we review existing
             % methods
% that address our problem and highlight their strengths and
% limitations.
In Section~ \ref{sec:existing} we present estimators from existing
work; in Section~\ref{sec:proposal} we discuss possible ways of
repairing the AIPW, and show that such repairs do not help us achieve
desirable properties such as asymptotic linearity. In
Section~\ref{sec:tmle} we present our proposed TML estimator an show
that it is asymptotically normal with known \textit{doubly robust
  asymptotic distribution}, where the latter concept means that the
distribution is known under consistent estimation of at least one
nuisance parameter. Simulation studies are presented in
Section~\ref{sec:simula}. These simulation studies demonstrate that
our estimators can lead to substantial bias reduction, as well as
improved coverage of the Wald-type confidence
intervals. Section~\ref{sec:discussion} presents some concluding
remarks and directions of future research.

\section{Illustrative Application}\label{sec:applica}
We illustrate our methods in the analysis of data from the ACTG 175
study \citep{hammer1996trial}. ACTG 175 was a randomized clinical
trial in which 2139 adults infected with the human immunodeficiency
virus type I, whose CD4 T-cell counts were between 200 and 500 per
cubic millimeter, were randomized to compare four antiretroviral
therapies: zidovudine (ZDV) alone, ZDV+didanosine(ddI),
ZDV+zalcitabine(ddC), and ddI alone.

One goal of the study was to compare the four treatment arms in terms
of the CD4 counts at week 96 after randomization. By week 96, 797
(37.2\%) subjects had dropped out of the study. Dropout rates varied
between 35.7-39.6\% across treatment arms. The investigators found
dropout to be associated to patient characteristics such as ethnicity
and history of injection-drug use, which are also associated with the
outcome, therefore causing informative missingness. Other baseline
variables collected at the beginning of the study include age, gender,
weight, CD4 count, hemophilia, homosexual activity, the Karnofsky
score, and prior antiretroviral therapy.

\subsection{Observed Data and Notation}
Let $W$ denote a vector of observed baseline variables, let $A$ denote
a binary treatment arm indicator (e.g., in our application we have
four such indicators). Let $Y$ denote the outcome of interest,
observed only when a missingness indicator $M$ is equal to
one. Throughout, we assume without loss of generality that $Y$ takes
values on $[0, 1]$.  We use the word \textit{model} in the classical
statistical sense to refer to a set of probability distributions for
the observed data $O=(W, A, M, MY)$. We assume that the true
distribution of $O$, denoted by $P_0$, is an element of the
nonparametric model, denoted by $\cal M$, and defined as the set of
all distributions of $O$ dominated by a measure of interest $\nu$. The
word \textit{estimator} is used to refer to a particular procedure or
method for obtaining estimates of $P_0$ or functionals of it. Assume
we observe an i.i.d. sample $O_1,\ldots,O_n$, and denote its empirical
distribution by $\Pn $. For a general distribution $P$ and a function
$f$, we use $Pf$ to denote $\int f(o)dP(o)$. We use $m(w)$ to denote
$E(Y\mid M=1,A=1,W=w)$, $g_A(w)$ to denote $P(A=a\mid W=w)$, and
$g_M(w)$ to denote $P(M=1\mid A=1,W=w)$. The index naught is added
when the expectation and probabilities are computed under $P_0$ (i.e.,
$m_0$, $g_{A,0}$, and $g_{M,0}$). We define $g(w)= g_A(w)g_M(w)$.

\subsection{Treatment Effect in Terms of Potential Outcomes and
  Identification}
Define the potential outcome $Y_1$ as the outcome that
would have been observed had study arm $A=1$ and missingness $M=1$
been externally set with probability one. The target estimand is
defined as $\thetac=E(Y_1)$. The index ``causal'' denotes a parameter
of the distribution of the potential outcome $Y_1$. We show that
$\thetac$ can be equivalently expressed as a parameter $\theta$ of the
observed data distribution $P_0(W,A,M,MY)$, under \ref{ass:cons}-4
below. This is useful since the potential outcome is not observed, in
contrast to the data vector $(W,A,M,MY)$, which we can make inferences
about. Define the following assumptions:
\begin{assumptioniden}[Consistency]\label{ass:cons}
  $Y=M\{AY_1 + (1-A)Y_0\}$,
\end{assumptioniden}
\begin{assumptioniden}[Randomization]\label{ass:random}
  $A$ is independent of $Y_1$ conditional on $W$,
\end{assumptioniden}
\begin{assumptioniden}[Missing at random]\label{ass:mar}
  $M$ is independent of $Y_1$ conditional on $(A,W)$,
\end{assumptioniden}
\begin{assumptioniden}[Positivity]\label{ass:pos}
  $g(w)>0$ with probability one over draws of $W$.
\end{assumptioniden}
\ref{ass:cons} connects the potential outcomes to the observed
outcome. \ref{ass:random} holds by design in a randomized trial such
as our illustrative example. \ref{ass:mar}, which is similar to that
in \cite{rubin1987multiple}, means that missingness is random within
strata of treatment and baseline variables (which is often abbreviated
as ``missing at random'', or MAR).  Equivalently, the MAR assumption
may be interpreted as the assumption that all common causes of
missingness and the outcome are observed and form part of the vector
of baseline variables $W$. \ref{ass:pos} guarantees that $m_0$ is well
defined.

Under \ref{ass:cons}-4 above, our target estimand $\thetac$ is
identified as $\theta_0 = E_{P_0}\{m_0(W)\}.$ Note that this parameter
definition allows us to compute the parameter value at any
distribution $P$ in the model $\mathcal M$. According to this
observation, we use the notation $\theta(P)=E_{P}\{m(W)\}$, where
$\theta_0=\theta(P_0)$.

\subsection{Data Analysis}
We present the results of applying our estimators to the ACTG data. To
estimate the probability of missingness conditional on baseline
variables $g_M$, we fit an ensemble predictor known as super learning
\citep{vanderLaan&Polley&Hubbard07, SL} to the missingness indicator
in each treatment arm. Super learning builds a convex combination of
predictors in a user-given library, where the combination weights are
chosen such that the cross-validated prediction risk is minimized. For
predicting probabilities, we define the prediction risk as the average
of the negative log-likelihood of a Bernoulli variable. The algorithms
used in the ensemble along with their weights are presented in
Table~\ref{tab:slcoef}. Note that the algorithms that more accurately
predict missingness are data-adaptive algorithms with flexible
functional forms, or algorithms that incorporate some type of variable
selection.

\begin{table}[!htb]
  \centering
  \begin{tabular}{l|cccc}
    \hline
    & \multicolumn{4}{c}{Treatment arm} \\
    Algorithm & ZVD & ZVD+ddI & ZVD+ddC & ddI \\
    \hline
    GLM & 0.00 & 0.00 & 0.00 & 0.00 \\
    Lasso & 0.02 & 0.21 & 0.00 & 0.85 \\
    Bayes GLM & 0.21 & 0.38 & 0.19 & 0.00 \\
    GAM & 0.00 & 0.00 & 0.02 & 0.00 \\
    MARS & 0.78 & 0.38 & 0.30 & 0.15 \\
    Random Forest & 0.00 & 0.03 & 0.49 & 0.00 \\\hline
  \end{tabular}
  \caption{Coefficients in the super learner convex combination for
    predicting 96 week dropout.}
  \label{tab:slcoef}
\end{table}

We also use the super learner to estimate the expected CD4 T-cell
count at 96 weeks after randomization among subjects still in the
study, conditional on covariates. The prediction risk in this case is
defined as the average of the squared prediction residuals. The
results are presented in Table~\ref{tab:slocoef}. For the outcome
regression, the best predictive algorithms are also data-adaptive.

\begin{table}[!htb]
  \centering
  \begin{tabular}{l|cccc}
    \hline
    & \multicolumn{4}{c}{Treatment arm} \\
    Algorithm & ZVD & ZVD+ddI & ZVD+ddC & ddI \\
    \hline
    GLM & 0.00 & 0.00 & 0.00 & 0.00 \\
    Lasso & 1.00 & 0.30 & 0.08 & 0.60 \\
    Bayes GLM & 0.00 & 0.02 & 0.00 & 0.00 \\
    GAM & 0.00 & 0.00 & 0.60 & 0.34 \\
    MARS & 0.00 & 0.00 & 0.00 & 0.06 \\
    Random Forest & 0.00 & 0.68 & 0.32 & 0.00 \\\hline
  \end{tabular}
  \caption{Coefficients in the super learner convex combination for
    predicting CD4 T-cell count.}
  \label{tab:slocoef}
\end{table}

The results in Tables~\ref{tab:slcoef} and \ref{tab:slocoef} highlight
the need to use data-adaptive estimators for the nuisance parameters
in the construction of a doubly robust estimator for $\theta_0$. As we
show below in Section~\ref{sec:existing}, standard doubly robust
estimators are not guaranteed to have desirable properties such as
$n^{1/2}$-consistency and doubly robust asymptotic linearity when such
data-adaptive estimators are used. This motivates the construction of
the estimators we propose.

Figure~\ref{fig:estima} shows the estimated CD4 T-cell count for each
treatment arm according to several estimators, along with their
corresponding 95\% confidence intervals. The targeted maximum
likelihood estimator \citep[TMLE][]{vanderLaanRose11} and the
augmented inverse-probability weighted estimator (AIPW) are standard
doubly robust estimators, whereas DTMLE and DAIPW are the
modifications described in Section~\ref{sec:proposal} below. % According
% to our findings in the next sections, these modified estimators are
% expected to have reduced biased under double misspecification
% (Lemma~\ref{lemma:bias}).
Unlike the TMLE and AIPW, the confidence
intervals of the DTMLE is expected to have correct asymptotic coverage
under consistent estimation of at least one nuisance parameter
(Theorem~\ref{theo:dr}). Unfortunately, the same claim does not seem
to hold for the DAIPW, although we expect this estimator to have
similar properties to the DTMLE in finite samples. For reference, we
also present the unadjusted estimate obtained by computing the
empirical mean of the outcome within each treatment arm among subjects
with observed outcomes.

\begin{figure}[!htb]
  \centering
  \includegraphics[scale = 0.3]{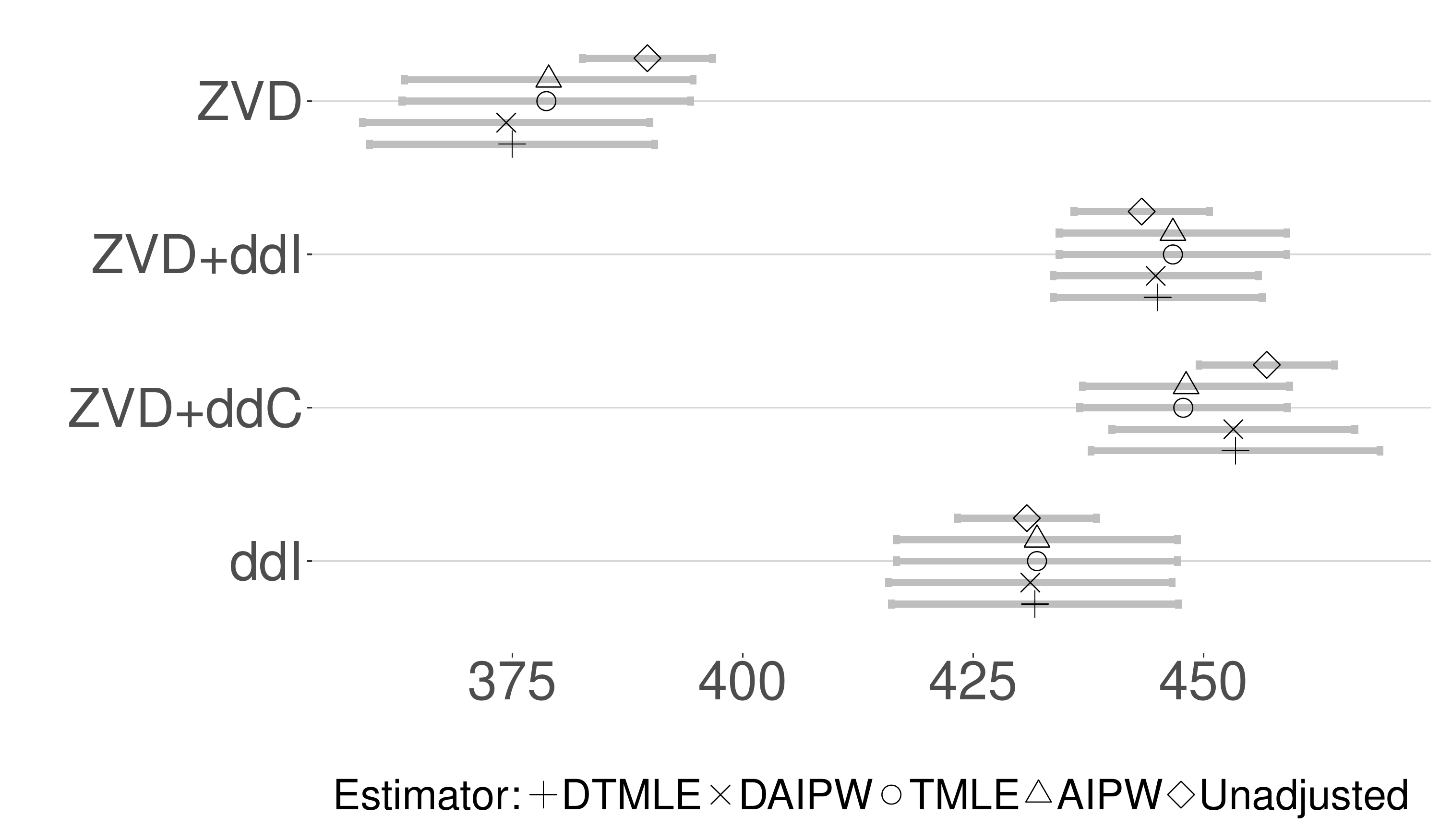}
  \caption{Estimated CD4 T-cell count on week 96 in each treatment arm, according
    to several estimators, along with confidence intervals.}
  \label{fig:estima}
\end{figure}

The dataset is available in the R package \texttt{speff2trial}
\citep{speff2trial}, the super learner predictor was computed using
the package \texttt{SuperLearner} \citep{SL}. R code to compute these
estimators is given in Appendix~\ref{sec:code}.

\section{Existing Estimators from the Semiparametric Efficiency Literature}\label{sec:existing}
We start by presenting the efficient influence function for estimation
of $\theta_0$ in model $\cal M$ \citep[see][]{hahn1998role}:
\begin{equation}
D_{\eta, \theta}(O) = \frac{A\,M}{g(W)}\{Y-m(W)\} + m(W) -
  \theta,\label{eq:defD}
\end{equation}
where we have denoted $\eta=(g, m)$. The efficient influence function
$D_{\eta,\theta}$ is a fundamental object for the analysis and
construction of estimators of $\theta_0$ in the non-parametric model
$\mathcal M$. First, it is a doubly robust estimating function, i.e.,
for given estimators $\hat m$ and $\hat g$ of $m_0$ and $g_0$,
respectively, an estimator that solves for $\theta$ in the following
estimating equation is consistent if at least one of $ m_0$ or $g_0$
is estimated consistently \citep[while the other converges to a limit
that may be incorrect, see Theorem 5.9 of][]{vanderVaart98}:
\begin{equation}
  \sum_{i=1}^n\frac{A_i\,M_i}{\hat g(W_i)}\{Y_i-\hat m(W_i)\} + \sum_{i=1}^n\left\{\hat m(W_i) -
    \theta\right\}=0.\label{eq:EE}
\end{equation}
The estimator constructed by directly solving for $\theta$ in the
above equation is often referred to as the augmented IPW estimator,
and we denote it by $\aipw$. Second, the efficient influence function
(\ref{eq:defD}) characterizes the efficiency bound for estimation of
$\theta_0$ in the model $\mathcal M$. Specifically, under consistent
estimation of $m_0$ and $g_0$ at a fast enough rate (which we define
below), an estimator that solves (\ref{eq:EE}) has variance smaller or
equal to that of any regular, asymptotically linear estimator of
$\theta_0$ in $\mathcal M$. This property is sometimes called
\textit{local efficiency}.

The augmented IPW has been criticized because directly solving the
estimating equation (\ref{eq:EE}) may drive the estimate out of bounds
of the parameter space \citep[see e.g.,][]{Gruber2010t}, which may
lead to poor performance in finite samples. Alternatives to repair the
AIPW have been discussed by \cite{Kang2007, Robins2007,
  tan2010bounded}. One such approach consists in solving the
estimating equation (\ref{eq:EE}) with the first term in the left hand
side divided by the empirical mean of the weights $A\,M/\hat
g(W)$. Alternatively, the targeted minimum loss based estimation
(TMLE) approach of \cite{vanderLaan&Rubin06, vanderLaanRose11}
provides a more principled method to construct estimators that stay
within natural bounds of the parameter space, for any smooth
parameter.

The TMLE of $\theta_0$ is defined as a substitution estimator
$\tmle=\theta(\tilde P)$, where $\tilde P$ is an estimate of $P_0$
constructed such that the corresponding $\tilde \eta$ and
$\theta(\tilde P)$ solve the estimating equation
$\sum_{i=1}^n D_{\tilde \eta, \theta(\tilde P)}(O_i)=0$. The estimator
$\tilde P$ is constructed by tilting an initial estimate $\hat P$
towards a solution of the relevant estimating equation, by means of a
maximum likelihood estimator in a parametric submodel.

Specifically, a TMLE may be constructed by fitting the logistic
regression model
\begin{equation}
\logit  m_\epsilon(w) = \logit \hat  m(w) + \epsilon \frac{1}{\hat
    g(w)},\label{eq:submodel}
\end{equation}
among observations with $(A_i,M_i)=(1,1)$. Here
$\logit(p)=\log\{p(1-p)^{-1}\}$. In this expression
$\epsilon$ is the parameter of the model, $\logit \hat  m(w)$ is an
offset variable, and the initial estimates $\hat  m$ and $\hat g$
are treated as fixed. The parameter $\epsilon$ is estimated using the
empirical risk minimizer
\[\hat \epsilon = \arg\max_{\epsilon}\sum_{i=1}^n A_iM_i \{Y_i\log
  m_\epsilon(W_i) + (1-Y_i)\log(1- m_\epsilon(W_i))\}.\] The tilted
estimator of $ m_0(w)$ is defined as
$\tilde m(w) = m_{\hat \epsilon}(w) = \expit\{\logit \hat m(w)+
\hat\epsilon / \hat g(w)\}$, where $\expit(x) = \logit^{-1}(x)$, and the TMLE of $\theta_0$ is defined as
\[\tmle=\frac{1}{n}\sum_{i=1}^n \tilde  m(W_i).\]
Because the empirical risk minimizer of model (\ref{eq:submodel})
solves the score equation
\[\sum_{i=1}^n\frac{A_i\,M_i}{\hat g(W_i)}\{Y_i -  m_{\hat
    \epsilon}(W_i)\}=0,\] it follows that
$\sum_{i=1}^n D_{\tilde \eta, \tmle}(O_i)=0$ with $\tilde \eta =
(\tilde g, \tilde  m)$. Since this procedure does not update the
estimator $\hat g$, we have $\tilde g=\hat g$.

Further discussion on the construction of the above TMLE may be found
in \cite{Gruber2010t}. \cite{Porter2011} provides an excellent review
of other doubly robust estimators along with a discussion of their
strengths and weaknesses. In this article we focus on the estimators
$\aipw$ and $\tmle$ defined above, but our methods can be used to
construct enhanced versions of other doubly robust estimators.

\subsection{Analysis of Asymptotic Properties of Doubly Robust Estimators}

The analysis of the asymptotic properties of the AIPW (as well as the
TMLE or any other estimator that solves the estimating equation
(\ref{eq:EE})) may be based on standard $M$-estimation and empirical
process theory. Here we focus on an analysis of the AIPW based on the
asymptotic theory presented in Chapter 5 of \cite{vanderVaart98}.

Define the following conditions:
\begin{assumption}[Doubly robust consistency]\label{ass:DR1}
  Let $||\cdot||$ denote the $L_2(P_0)$ norm defined as
  $||f||^2=\int f^2 dP_0$. Assume
  \begin{enumerate}[label=(\roman*)]
\item There exists $\eta_1=(g_1, m_1)$
  with either $g_1 = g_0$ or $ m_1= m_0$ such that
  $||\hat m - m_1||=o_P(1)$ and $||\hat g - g_1||=o_P(1)$.
\item For $\eta_1$ as above, $||\hat m - m_1||\,||\hat g -
  g_1||=o_P(n^{-1/2})$.
\end{enumerate}
\end{assumption}
\begin{assumption}[Donsker]\label{ass:donsker}
  Let $\eta_1$ be as in~\ref{ass:DR1}-(i). Assume the class of
  functions $\{\eta=(g, m):|| m- m_1||<\delta, ||g-g_1||<\delta\}$ is
  Donsker for some $\delta >0$.
\end{assumption}
Under \ref{ass:DR1}-(i) and 2, a straightforward
application of Theorems 5.9 and 5.31 of \cite{vanderVaart98} \citep[see also
example 2.10.10 of][]{vanderVaart&Wellner96} yields
\begin{equation}
  \aipw-\theta_0= \beta(\hat\eta) +
  (\Pn - P_0)D_{\eta_1, \theta_0} + o_P\big(n^{-1/2} + |\beta(\hat\eta)|\big),\label{eq:wh}
\end{equation}
where $\beta(\hat\eta) = P_0 D_{\hat\eta, \theta_0}$. Thus, the
probability distribution of doubly robust estimators depends
on $\hat \eta$ through the ``drift'' term $\beta(\hat\eta)$. For our
parameter $\theta$ the drift term is given by
\begin{equation}
\beta(\hat\eta) = \int
  \frac{1}{\hat g}(\hat g - g_0)(\hat m - m_0)dP_0.\label{eq:defbeta}
\end{equation} Note that under
\ref{ass:DR1}, $\beta(\hat\eta)$ converges to zero in probability so
that $\aipw$ and $\tmle$ are consistent. Efficiency under
$\eta_1=\eta_0$ can be proved as follows. The Cauchy-Schwartz
inequality shows that
\[\beta(\hat\eta)\leq C||\hat  m -  m_0||\,||\hat g - g_0||,\]
for some constant $C$. Under \ref{ass:DR1} and $\eta_1=\eta_0$, we get
$\beta(\hat\eta) = o_P(n^{-1/2})$ so that (\ref{eq:wh}) yields
\begin{equation*}
  \aipw - \theta_0=(\Pn - P_0)D_{\eta_0, \theta_0} + o_P\big(n^{-1/2}\big).\label{eq:tmleeff}
\end{equation*}
An identical result holds replacing $\aipw$ by $\tmle$ in the above
display. Asymptotic normality and efficiency follows from the central
limit theorem.

In the more common doubly robust scenario in which at most one of
$ m_0$ or $g_0$ is consistently estimated, the large sample analysis
of doubly robust estimators relies on the assumption that
$\beta(\tilde\eta)$ is asymptotically linear \cite[see Appendix 18
of][]{vanderLaanRose11}. If $\hat \eta$ is estimated in a parametric
model, the delta method yields the required asymptotic
linearity. However, this assumption is hard to verify when $\hat\eta$
uses data-adaptive estimators; in fact there is no reason to expect
that it would hold in general.

In the remainder of the paper we construct drift-corrected estimators
$\daipw$ and $\dtmle$ that control the asymptotic behavior through
estimation of the drift term in the more plausible doubly robust
situation where either $g_1 = g_0$ or $m_1= m_0$, but not necessarily
both.

\begin{remark}[Asymptotic bias of the AIPW and TMLE under double
  inconsistency]\label{lemma:bias}
  Assume $\hat \eta = (\hat g, \hat \eta)$ converges to some
  $\eta_1=(g_1, m_1)$. Define $\theta_1=P_0 m_1$, and note that
  $D_{\eta_1,\theta_1}=D_{\eta_1,\theta_0}-\theta_1 + \theta_0$. Under
  \ref{ass:donsker}, an application of Theorem 5.31 of
  \cite{vanderVaart98} yields
  \begin{equation*}
    \aipw-\theta_1= \beta(\hat \eta) +
    (\Pn - P_0)D_{\eta_1, \theta_1} + o_P\big(n^{-1/2} + |\beta(\hat\eta)|\big).
  \end{equation*}
  Substituting $D_{\eta_1,\theta_1}=D_{\eta_1,\theta_0}-\theta_1 +
  \theta_0$ yields
  \begin{equation*}
    \aipw-\theta_0= \beta(\hat \eta) +
    (\Pn - P_0)D_{\eta_1, \theta_0} + o_P\big(n^{-1/2} + |\beta(\hat\eta)|\big).
  \end{equation*}
  The above expression also holds for $\aipw$ replaced with $\tmle$
  and $\hat\eta$ replaced with $\tilde\eta$. The empirical process
  term $(\Pn - P_0)D_{\eta_1, \theta_0}$ has mean zero. Thus,
  controlling the magnitude of $\beta(\hat\eta)$ and
  $\beta(\tilde\eta)$ is expected to reduce the bias of $\aipw$ and
  $\tmle$, respectively, in the double inconsistency case in which
  $m_1\neq m_0$ and $g_1\neq g_0$.
\end{remark}

\section{Repairing the AIPW Estimator Through Estimation of
  $\beta(\hat\eta)$}\label{sec:proposal}

As seen from the analysis of the previous section, the consistency
\ref{ass:DR1} with $\eta_1=\eta_0$ is key in proving the optimality
($n^{1/2}$-consistency, asymptotic normality, efficiency) of doubly
robust estimators such as the TMLE and the AIPW. The asymptotic
distribution of doubly robust estimators under violations of this
condition depends on the behavior of the drift term
$\beta(\hat\eta)$. % For example, if $\hat\eta$ is the MLE in a
% parametric model, the Delta method yields the asymptotic linearity of
% $\beta(\hat\eta)$, and therefore $\tmle$ is also asymptotically
% linear. In addition,
We propose a method that controls the asymptotic behavior of
$\beta(\hat\eta)$. This is achieved through a decomposition into score
functions associated to estimation of $m_0$ and $g_0$. In light of
Remark~\ref{lemma:bias} controlling the magnitude and variation of
$\beta(\hat\eta)$ is also important to reduce the bias of the TMLE
when either $g_0$ or $ m_0$ are inconsistently estimated.
% Situations in which
% $\hat\eta$ is the result of data-adaptive estimation methods in
% which the analysis of $\beta(\hat\eta)$ is not straightforward.

% Specifically, under \ref{ass:DR1}, $\beta(\hat\eta)$ may be decomposed as
% \begin{equation}\beta(\hat\eta) =  \beta_g(\hat g) +
%   \beta_m(\hat m) + o_P(n^{-1/2}),\label{eq:biasdec}
% \end{equation}
% where
% \begin{align*}
%   \beta_g(\hat g) &=  \int \frac{1}{\hat g}\{g_0-\hat g\}\{ m_0- m_1\}dP_0\\
% \beta_m(\hat m)&= \int \frac{1}{g_1}\{g_0-g_1\}\{ m_0-\hat m\}dP_0.
% \end{align*}
% Our proposal is focused on estimation of the drift terms
% $\beta_g(\hat g)$ and $\beta_m(\hat m)$. These drift terms depend
% on the differences $ m_0- m_1$ and $g_0-g_1$, which are difficult
% to estimate in their current form.

We introduce the following strengthened doubly robust consistency
condition:
\begin{assumption}[Strengthened doubly robust consistency]\label{ass:DR2}
  $\hat \eta = (\hat g, \hat \eta)$ converges to some
  $\eta_1=(g_1, m_1)$ in the sense that
  $||\hat  m -  m_1||=o_P(n^{-1/4})$ and $||\hat g - g_1||=o_P(n^{-1/4})$ with
  either $g_1 = g_0$ or $ m_1= m_0$.
\end{assumption}

The following lemma provides an approximation for the drift term in
terms of score function in the tangent space of each of the models for
$g_0$ and $m_0$. Such approximation is achieved through the definition
of the following univariate regression functions:
\begin{align}
  \gamma_{A,0}(W) &= P_0\big\{A = 1 \mid m_1(W)\big\},\notag\\
  \gamma_{M,0}(W) &= P_0\big\{M = 1\mid A = 1, m_1(W)\big\},\notag\\
  r_{A,0}(W) &= E_{P_0}\left\{\frac{A - g_{A,1}(W)}{g_{A,1}(W)}\mathrel{\bigg|}
               m_1(W)\right\}\label{eq:regs},\\
  r_{M,0}(W) &= E_{P_0}\left\{\frac{M - g_{M,1}(W)}{g_1(W)}\mathrel{\bigg|} A=1, m_1(W)\right\}\notag,\\
  e_0(W) &= E_{P_0}\big\{Y - m_1(W)\mid A=1,M=1,g_1(W)\big\}\notag.
\end{align}
Note that the residual regressions $r_{A,0}$, $r_{M,0}$, and $e_0$ are
equal to zero if the limits $g_{A,1}$, $g_{M,1}$, and $m_1$ of the
nuisance estimators are correct. To see this, it suffices to replace
$g_{A,0}$ for $g_{A,1}$ in $r_{A,0}$, and apply the iterated
expectation rule conditioning first on $W$.

\begin{theorem}[Asymptotic approximation of the drift term]\label{lemma:betarep}
 % Define the auxiliary covariates
 %  \[H_1(w)= \frac{ m_1(w)}{g_{A,1}(w)};\quad    L_1(w)= \frac{ m_1(w)}{g_{M,1}(w)g_{A,1}(w)}\]
% \hat H(w)= \frac{\hat m(w)}{\hat g_A(w)};\\
%     \hat L_1(w)&=\frac{\hat m(w)}{g_1(w)};\quad
%                    \hat L(w) =\frac{\hat m(w)}{\hat g(w)}.
% and the residuals
% \begin{align*}
%   r_{A,0}(W) &= \gamma_{A,0}(W) - g_{A,1}(W)\\
%   r_{M,0}(W) &= \gamma_{M,0}(W) - g_{M,1}(W)\\
% e_0(W)&=\mu_0(W) - m_1(W).
% \end{align*}
Denote $\lambda_0=(\gamma_{A,0}, \gamma_{M,0}, r_{A,0}, r_{M,0},
e_0)$, and define the following score functions:
\begin{align*}
  D_{Y,\hat m, \lambda_0}(O) &= A\,M\left\{\frac{r_{A,0}(W)}{\gamma_{A,0}(W)}
                           + \frac{r_{M,0}(W)}{\gamma_0(W)}\right\}\{Y - \hat m(W)\}\\
  D_{M, \hat g, \lambda_0}(O) &= \frac{A\,e_0(W)}{\hat
                         g(W)}\{M - \hat g_M(W)\}\\
  D_{A, \hat g, \lambda_0}(O) &= \frac{e_0(W)}{\hat g_A(W)}\{A - \hat g_A(W)\},
\end{align*}
where $\gamma_0(w)=\gamma_{A,0}(w)\gamma_{M,0}(w)$. Under \ref{ass:DR2} we have
$\beta(\hat\eta)=P_0\{D_{A,\hat g, \lambda_0} + D_{M,\hat g, \lambda_0} +
D_{Y,\hat m, \lambda_0}\} + o_P(n^{-1/2})$.
\end{theorem}
% \begin{proof}
%   See Supplementary Materials.
% \end{proof}

Unlike expression~\ref{eq:defbeta}, the above approximation of the
drift depends only on one-dimensional nuisance parameters which are
easily estimable through non-parametric smoothing techniques.  These
one-dimensional parameters are functions of the possibly misspecified
limits of your estimators. However, in what follows this does not
prove to be problematic. In particular, $\beta(\hat\eta)$ may be
estimated as follows. First, we construct an estimator of $\lambda_0$
component-wise by fitting non-parametric regression estimators. Since
all the regression functions in (\ref{eq:regs}) are one-dimensional,
they may be estimated by fitting a kernel regression. For instance,
for a second-order kernel function $K_h$ with bandwidth $h$ the
estimator of $e_0$ is given by
\begin{equation}
  \hat e(w) = \frac{\sum_{i = 1}^nA_i\,M_i\,K_{\hat h}\{\hat g(W_i) -
    \hat g(w)\}\{Y_i - \hat m(W_i)\}}{\sum_{i=1}^nA_i\,M_i\,K_{\hat h}\{\hat g(W_i) -
    \hat g(w)\}}.\label{eq:rhat}
\end{equation}
The bandwidth is chosen as $\hat h= n^{-0.1}\hopt$, where
$\hopt$ is the optimal bandwidth chosen using
K-fold cross-validation \citep[the optimality of this selector is
discussed in][]{vanderVaart&Dudoit&vanderLaan06}.  This bandwidth
yields a convergence rate that allows application of uniform central
limit theorems \citep[see Theorems 4 and 5 of][]{gine2008uniform}. % The convergence rate
% of $\hat h$ is $n^{-3/10}$, which allows while maintaining a
% convergence rate of at least $n^{-1/4}$ for $K_{\hat h}$ in $L_2(P_0)$
% norm.

An estimator of the drift term may be constructed as
\begin{multline}\hat \beta(\hat \eta) = \frac{1}{n}\sum_{i=1}^n\left[\frac{\hat e(W_i)}{\hat g_A(W_i)}\{A_i-\hat g_A(W_i)\}
    +\frac{A_i\,\hat e(W_i)}{\hat g(W_i)}\{M_i-\hat g_M(W_i)\} +
  \right.\\ \left. A_iM_i\left\{\frac{\hat r_A(W_i)}{\hat\gamma(W_i)}
      + \frac{\hat r_M(W_i)}{\hat\gamma_M(W_i)}\right\}\{Y_i - \hat
    m(W_i)\}\right].\label{eq:hatbeta}
\end{multline}
In light of equation~(\ref{eq:wh}), the above estimator may be
subtracted from the AIPW (or the TMLE) to obtain a drift-corrected
estimator. We denote this estimator by
$\daipw=\aipw-\hat\beta(\hat\eta)$. %The following theorem
% establishes the asymptotic behavior if this estimator.

% \begin{theorem}[Doubly Robust Asymptotic Distribution of $\daipw$]\label{theo:draipw}
%   Assume \ref{ass:donsker} and
%   \ref{ass:DR2}. Then
%   \[n^{1/2}(\dtmle - \theta_0)\to N(0, \sigmadaipw^2),\]
%   where $\sigmadaipw^2 = \var\{\Psidaipw(O)\}$ and $\Psidaipw(O)=D_{\eta_1, \theta_0}(O) -
%   D_{Y,m_1,\lambda_0}(O) - D_{M,g_1,\lambda_0}(O) - D_{A,g_1,\lambda_0}(O)$.
% \end{theorem}

Though sensible in principle, $\daipw$ suffers from drawbacks similar
to the standard AIPW estimator $\aipw$: it may yield an estimator out
of bounds of the parameter space and therefore have suboptimal finite
sample performance (we illustrate this in our simulation study in
Section~\ref{sec:simula}). In addition, a large sample analysis of
$\daipw$ suggests that the $n^{1/2}$-consistency of $\daipw$ requires
consistent estimation of $\lambda_0$ at the $n^{1/2}$ parametric
rate. In particular, under \ref{ass:DR1}-2, equation~(\ref{eq:wh})
yields
\begin{equation}
\label{eq:wh1}
\daipw-\theta_0= \beta(\hat\eta) - \hat\beta(\hat\eta)+
  (\Pn - P_0)D_{\eta_1, \theta_0} + o_P\big(n^{-1/2} +
  |\beta(\hat\eta)|\big).
\end{equation}
Lemma~\ref{lemma:asbeta} in the appendix shows that, under \ref{ass:DR2},
\begin{equation}
\beta(\hat\eta) - \hat\beta(\hat\eta) = -(\Pn-P_0)\{D_{A,\hat g, \lambda_0} + D_{M,\hat g, \lambda_0} +
  D_{Y,\hat m, \lambda_0}\} + o_P(n^{-1/2}).\label{eq:wh2}
\end{equation} Asymptotic linearity of
$\daipw$ would then require that $|\beta(\hat\eta)|=O_P(n^{-1/2})$, so that
the last term in the right-hand side of expression~(\ref{eq:wh1}) is
$o_P(n^{-1/2})$. This would require $\lambda_0$ to be estimated at
rate $n^{1/2}$, which is in general not achievable in the
non-parametric model (e.g., the convergence rate of a kernel
regression estimator with second order kernel and optimal bandwidth is
$n^{2/5}$). It would thus appear that the $\daipw$ estimator will not
generally be asymptotically linear if the estimator of $\lambda_0$
converges to zero more slowly than $n^{-1/2}$.

Surprisingly, the large-sample analysis of the $\dtmle$ counterpart
presented in Section~\ref{sec:tmle} below requires slower convergence
rates for the estimator of $\lambda_0$, such that a Kernel regression
estimator provides a sufficiently fast rate. This fact has been
previously noticed in the context of estimation of a counterfactual
mean by \cite{benkeser2016doubly}. %Unlike \cite{benkeser2016doubly},
We note that the optimal bandwidth $\hat h_{\text{opt}}$ in estimation
of $\lambda_0$ yields estimators for which uniform central limit
theorems do not apply. Therefore we propose to undersmooth using the
bandwidth $\hat h$.

\section{Targeted Maximum Likelihood Estimation with
  Doubly Robust Inference}\label{sec:tmle}

As transpires from the developments of the previous section, it is
necessary to construct estimators $\hat\eta$ such that
$\beta(\hat\eta)$ is $O_P(n^{-1/2})$. In light of
expression~(\ref{eq:wh2}), this can be achieved through the
construction of an estimator $\tilde\eta$ that satisfies
$\hat\beta(\tilde\eta)=0$. This construction is based on the fact that
$D_{Y,\hat m, \lambda_0}$, $D_{M,\hat g, \lambda_0}$, and
$D_{M,\hat g, \lambda_0}$ are score equations in the model for $m_0$,
$g_{M,0}$, and $g_{A,0}$, respectively. As a result, adding the
corresponding covariates to a logistic tilting model will tilt an
initial estimator $\hat \eta=(\hat g,\hat m)$ towards a solution
$\tilde \eta$ of the bias-reducing estimating equations
$\hat \beta(\tilde \eta)=0$, in a similar way to the logistic tilting
submodel (\ref{eq:submodel}).

The proposed drift-corrected TMLE is defined by the following algorithm:
\begin{enumerate}[label = Step~\arabic*., align=left, leftmargin=*]
\item \textit{Initial estimators.} Obtain initial estimators
  $\hat g_A$, $\hat g_M$, and $\hat m$ of $g_{A,0}$, $g_{M,0}$, and
  $ m_0$. These estimators may be based on data-adaptive predictive
  methods that allow flexibility in the specification of the
  corresponding functional forms. Construct estimators $\hat\gamma_A$,
  $\hat\gamma_M$, $\hat\mu$ of $\gamma_{A,0}$, $\gamma_{M,0}$,
  $\mu_0$, respectively, by fitting kernel regression estimators as
  described in the previous subsection.
\item \textit{Compute auxiliary covariates.} For each
  subject, compute the auxiliary covariates
\[W_1(w)=\frac{1}{\hat g(w)},\,
    W_2(w)=\frac{\hat r_A(w)}{\hat \gamma(w)} + \frac{\hat
      r_M(w)}{\hat \gamma_M(w)},\,
    Z_A(w)=\frac{\hat e(w)}{\hat g_A(w)},\, Z_M(w)=\frac{\hat e(w)}{\hat g(w)}
\]
\item \textit{Solve estimating equations.} Estimate the parameter
  $\epsilon = (\epsilon_A, \epsilon_M, \epsilon_{Y,1}, \epsilon_{Y,2})$ in the logistic tilting models
  \begin{align}
    \logit  m_\epsilon(w) &= \logit \hat m(w)  + \epsilon_{Y,1} W_1(w) +
                             \epsilon_{Y,2} W_2(w),\label{eq:submodelY}\\
    \logit g_{M,\epsilon}(w) &= \logit \hat g_M(w) + \epsilon_M
                               Z_M(w).\label{eq:submodelT}\\
    \logit g_{A,\epsilon}(w) &= \logit \hat g_A(w) + \epsilon_A
                               Z_A(w)\label{eq:submodelT}
  \end{align}
  Here, $\logit \hat m(w)$, $\logit \hat g_A(w)$, and
  $\logit \hat g_M(w)$ are offset variables (i.e., variables with
  known parameter equal to one). The above parameters may be estimated
  by fitting standard logistic regression models. For example,
  $(\epsilon_{Y,1}, \epsilon_{Y,2})$ may be estimated through a
  logistic regression model of $Y$ on $(W_1,W_2)$, with no intercept
  and with offset $\logit \hat m(W)$ among observations with
  $(A,M)=(1,1)$. Likewise, $\epsilon_M$ is estimated through a
  logistic regression model of $M$ on $Z_M$ with no intercept and an
  offset term equal to $\logit \hat g_M(W)$ among observations with
  $A=1$. Lastly, $\epsilon_A$ may be estimated by fitting a logistic
  regression model of $A$ on $Z_A$ with no intercept and an offset
  term equal to $\logit \hat g_A(W)$ using all observations. Let
  $\hat\epsilon$ denote these estimates.
\item \textit{Update estimators and iterate.} Define the updated
  estimators as $\hat  m =  m_{\hat \epsilon}$,
  $\hat g_M = g_{M,\hat\epsilon}$, and $\hat g_A=g_{A,\hat
    \epsilon}$. Repeat steps 2-4 until convergence. In practice, we
  stop the iteration once
  $\max\{|\hat\epsilon_A|, |\hat\epsilon_M|, |\hat\epsilon_{Y,1}|,
  |\hat\epsilon_{Y,2}|\}< 10^{-4}n^{-3/5}$.
\item \textit{Compute TMLE.} Denote the estimators in the last step of
  the iteration with $\tilde m$, $\tilde g_M$, and $\tilde g_M$. The
  drift-corrected TMLE of $\theta_0$ is defined as
\[\dtmle = \frac{1}{n}\sum_{i=1}^n \tilde m(W_i).\]
\end{enumerate}

The large sample distribution of the above TMLE is given in the
following theorem:

\begin{theorem}[Doubly Robust Asymptotic Distribution of $\dtmle$]\label{theo:dr}
Assume \ref{ass:donsker} and
  \ref{ass:DR2} hold for $\tilde\eta$, and denote the limit of
  $\tilde \eta$ with $\eta_1$. Then
  \[n^{1/2}(\tmle - \theta_0)\to N(0, \sigma^2),\]where
  $\sigma^2 = \var\{D_{\dr}(O)\}$ and
  $D_{\dr}(O)=D_{\eta_1, \theta_0}(O) - D_{Y,m_1,\lambda_0}(O) -
  D_{M,g_1,\lambda_0}(O) - D_{A,g_1,\lambda_0}(O)$.
\end{theorem}

Note that, in an abuse of notation, we have denoted the limit of
$\tilde\eta$ with $\eta_1$, though this limit need not be equal to the
limit of the initial estimator $\hat\eta$.

\ref{ass:DR2}, assumed in the previous theorem, is stronger than the
standard double robustness \ref{ass:DR1}. Under \ref{ass:DR1},
$\tilde m$ or $\tilde g$ may converge to their misspecified limits
arbitrarily slowly as long as the product of their $L_2(P_0)$ norms
converges at rate $n^{1/2}$. Under \ref{ass:DR2} each estimator is
required to converge to its misspecified limit at rate $n^{1/4}$. This
is a mildly stronger condition that we conjecture is satisfied by many
data-adaptive prediction algorithms. In particular, it is satisfied by
empirical risk minimizers (minimizing squared error loss or quasi
log-likelihood loss) over Donsker classes. An example of a
data-adaptive estimator that satisfies \ref{ass:DR2} is the highly
adaptive lasso (HAL) proposed by
\cite{van2015generally}. \ref{ass:DR2} is necessary to control the
convergence rate of the estimator $\hat\lambda$. The reader interested
in the technical details is encouraged to consult the proof of the
theorem in the Supplementary Materials.

In light of Theorem~\ref{theo:dr}, the Wald-type confidence interval
$\dtmle \pm z_{\alpha} \hat\sigma/\sqrt{n}$, where $\hat\sigma^2$ is
the empirical variance of
$\hat D_{\dr}(O)=D_{\tilde \eta, \dtmle}(O) - D_{Y,\tilde m,\hat
  \lambda}(O) - D_{M,\tilde g, \hat \lambda}(O) - D_{A,\tilde g, \hat
  \lambda}(O)$ has correct asymptotic coverage $(1-\alpha)100\%$,
whenever at least one of $\tilde g$ and $\tilde m$ converges to its
true value at the stated rate. However, computation of the confidence
interval does not require one to know which of these nuisance
parameters is consistently estimated.

% \begin{proof}
% See Supplementary Materials.
% \end{proof}

% \subsection{Bias Reduction Under Misspecification of Both Models}\label{sec:bias}

% \begin{remark}[Ensuring positivity of $\tilde g$]
% \end{remark}

\section{Simulation Studies}\label{sec:simula}

We compare the performance of our proposed enhanced estimators
$\dtmle$ and $\daipw$ with their standard versions $\tmle$ and
$\aipw$, using the following data distribution:
\begin{align*}
  \logit g_{M,0}(a,w)=&\,2 -w_1+4w_2-2w_4+3w_2w_6 + 3w_1w_5w_6 -\\
&\,a(1.5-4w_1+4w_2+2w_3-7w_1w_2-3w_2w_4w_5)\\
\logit  m_0(a,w)=&\,-0.5 -w_1-w_2+w_4+2w_2w_6 + 2w_1w_5w_6 -\\
                      &\,a(2-w_1+3w_2+w_3-6w_1w_2-4w_2w_4w_5).
\end{align*}
For exogenous variables $\varepsilon_1,\ldots,\varepsilon_6$ distributed independently as uniform
variables in the interval $(0,1)$, $W_1,\ldots,W_6$ were generated as
\begin{align*}
  W_1 &= \log(\varepsilon_1 + 1)\\
  W_2 &= \varepsilon_2 / (1 + \varepsilon_1^2)\\
  W_3 &= \varepsilon_1 + 1 / (\varepsilon_3 + 1)\\
  W_4 &= \sqrt{\varepsilon_2 + \varepsilon_4}\\
  W_5 &= \varepsilon_5\varepsilon_4\\
  W_6 &= 1 / (\varepsilon_2 + \varepsilon_6 + 1).
\end{align*}
The treatment probabilities were set to $g_{A,0}(w)=0.5$, corresponding
with a randomized trial with equal allocation, and the outcome was
generated as $Y\mid \{A=a,W=w\}\sim \text{Bernoulli}(m_0(a,w))$. For
this data generating mechanism we have a treatment effect of
$\theta_0\approx 0.2328$, and
$E(Y\mid A=1,M=1) - E(Y\mid A=0,M=1)\approx 0.3258$, indicating a
strong selection bias due to informative missingness.

For each sample size $n$ in the grid $\{200,800,1800,3200,5000,7200,9800\}$,
we generate 1000 datasets with the above distribution, and test four
different scenarios for estimation of $g_{M,0}$ and $m_0$:
\begin{enumerate*}[label=(\alph*)]
\item consistent estimation of both $g_{M,0}$ and $m_0$,
\item consistent estimation of $m_0$ and inconsistent estimation of
  $g_{M,0}$,
\item consistent estimation of $g_{M,0}$ and inconsistent estimation
  of $m_0$, and
\item inconsistent estimation of both $g_{M,0}$ and $m_0$.
\end{enumerate*}

Consistent estimators of $g_{M,0}$ and $m_0$ are constructed by first
creating a model matrix containing all possible interactions of $W$ up
to fourth order, and then running $L_1$ regularized logistic
regression. Inconsistent estimation follows the standard practice of
fitting logistic regression models on main terms only. The use of
$L_1$ regularization provides an example in which the asymptotic
linearity of the drift term is not guaranteed. Since we do not assume
we know which interactions are present, the use of data-adaptive
estimators is the only possible way to obtain consistent estimators,
as it is in most real data applications.

In all scenarios, the treatment mechanism is consistently estimated
by fitting a logistic regression of $A$ on $W$ including main terms
only, even though $g_{A,0}$ is known by design.  Intuitively, the
purpose of this model fit is to capture chance imbalances of the
baseline variables $W$ between study arms for a given data set; these
imbalances can then be adjusted to improve efficiency. The general
theory underlying efficiency improvements through estimation of known
nuisance parameters such as $g_A$ is presented, e.g., by
\cite{Robins&Rotnitzky&Zhao94} and \cite{vanderLaan2003}.

We compare the performance of the four estimators in terms of four
metrics:

\begin{enumerate}[label=(\roman*)]
\item Coverage probability of a confidence interval based on the
  central limit theorem, with variance estimated
  as \[\hat\sigma^2=\frac{1}{n}\sum_{i=1}^n\text{IF}^2(O_i),\] where
  IF is the estimated influence function of the corresponding
  estimator. For $\aipw$ and $\tmle$, the influence function used is
  the efficient influence function $D_{\eta,\theta}$. For $\daipw$ and
  $\dtmle$, the influence function $D_{\dr}$ given in
  Theorem~\ref{theo:dr}.

  Confidence intervals for $\aipw$ and $\tmle$ are expected to have
  correct coverage in scenario (a), incorrect coverage in scenario
  (b), and conservative coverage in scenario (c). In light of
  Theorem~\ref{theo:dr}, the confidence interval based on $\dtmle$ is
  expected to have correct coverage in scenarios (a)-(c). The behavior
  of the confidence interval based on $\daipw$ is conjectured to have
  similar performance to the $\dtmle$, but our theory does not show
  this in general.
\item The absolute value of the bias scaled by $\sqrt{n}$. This value
  is expected to converge to zero in scenarios (a)-(c) for all
  estimators, and to diverge in scenario (d). For scenario (d), in
  light of Remark~\ref{lemma:bias}, we conjecture that $\daipw$ and
  $\dtmle$ have generally smaller bias than $\aipw$ and $\tmle$,
  respectively.
\item The squared root of the relative MSE (RMSE), scaled by
  $\sqrt{n}$. The RMSE is defined as the MSE divided by the efficiency
  bound $\var\{D_{\eta_0,\theta_0}(O)\}$. This metric is expected to
  converge to one for all estimators in scenario (a) (i.e., all
  estimators are efficient), it is expected
  to converge to some other value in scenarios (b)-(c), and it is
  expected to diverge in scenario (d).
\item The average of the estimated standard deviations $\hat\sigma$
  across 1000 datasets divided by the standard deviation of the
  estimates $\hat\theta$. This metric is expected to converge to one
  for all estimators in scenario (a), and for estimators $\daipw$ and
  $\dtmle$ in scenarios (b)-(c).
\end{enumerate}

\begin{figure}[!htb]
  \centering
  \includegraphics[scale = 0.4]{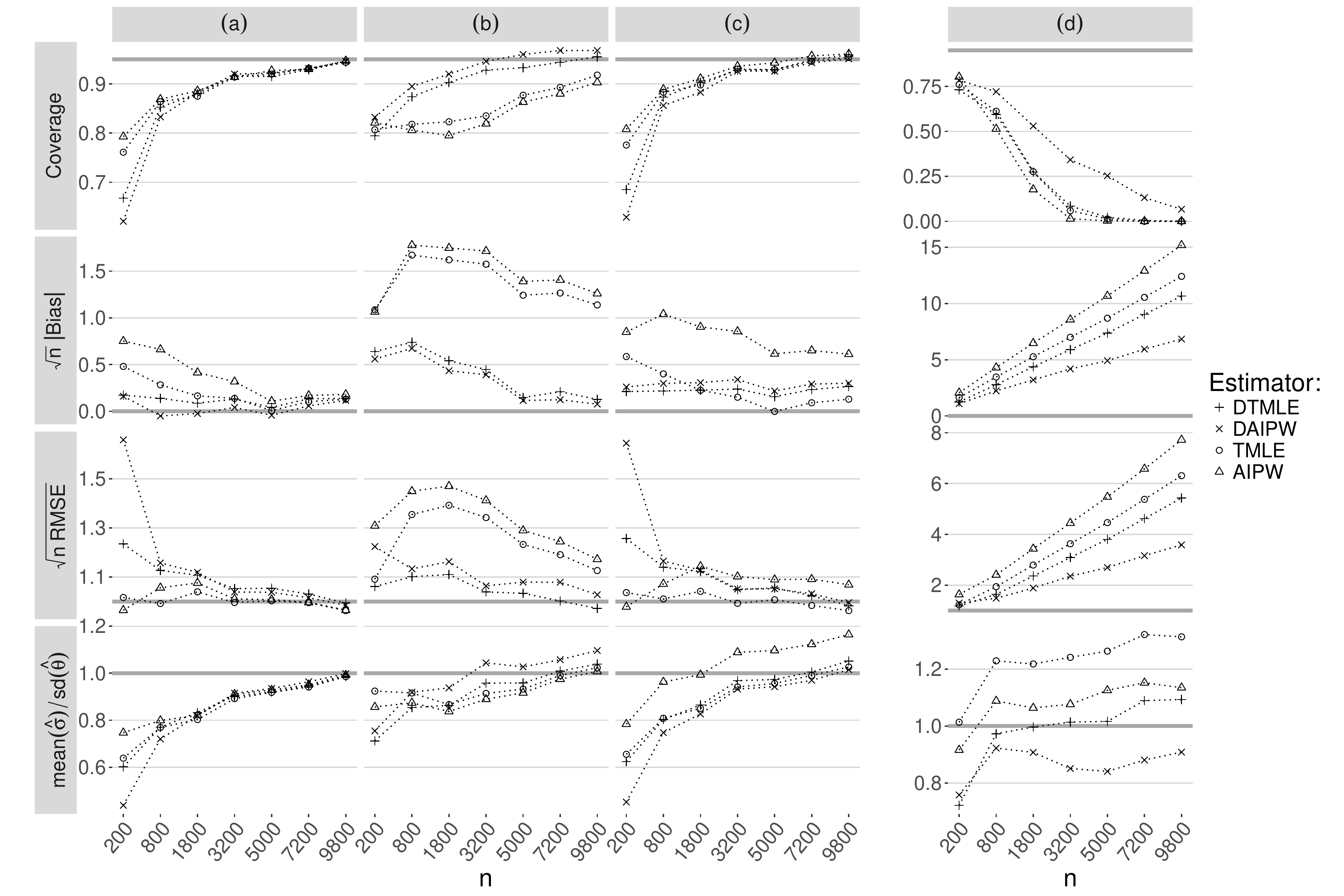}
  \caption{Results of the simulation study.}
    \label{fig:res}
\end{figure}
The results of the simulation are presented in
Figure~\ref{fig:res}. In addition to corroborating the expected
attributes of the estimators outlined in (i)-(iv) above, the following
characteristics deserve further observation:

\begin{itemize}
\item $\daipw$ has a much higher variance compared to all other
  estimators in scenario (a) for small samples ($n=200$) . This is
  possibly a consequence of inverse weighting by small probabilities
  in the definition of the correction factor $\hat\beta(\tilde\eta)$
  (see equation \ref{eq:hatbeta}). This also affects $\dtmle$, but to
  a lesser extent.
\item $\daipw$ and $\dtmle$ have considerably better performance than
  $\aipw$ and $\tmle$ in scenario (b): they achieve the asymptotic
  efficiency bound and have significantly smaller bias.
\item $\daipw$ has smaller bias than all competitors under
  scenario (d).
\end{itemize}

\section{Concluding Remarks}\label{sec:discussion}
We present estimators of the effect of treatment in randomized trials
with missing outcomes, where the outcomes are missing at random. One
of our proposed estimators, the DTMLE, is CAN under data-adaptive
estimation of the missingness probabilities and the outcome
regression, under consistency of at least one of these estimators. We
present the doubly-robust influence function of the estimator, which
can be used to construct asymptotically valid Wald-type confidence
intervals. We show that the implied asymptotic distribution is valid
under a smaller set of assumptions, compared to existing estimators.

As an anonymous referee pointed out, the method of
\cite{benkeser2016doubly} could be applied to our problem by defining
$T=AM$ and estimating $E\{E(Y\mid T=1, W)\}$. We find this approach
unsatisfactory because it ignores intrinsic properties of the
variables $A$ and $M$, which are more appropriately exploited when
modeled independently. For example, $P(A=1\mid W)$ is known in a
randomized trial, and a logistic regression model with at least an
intercept term provides a consistent estimator. Furthermore, covariate
adjustment through such logistic model is known to improve the
precision of the resulting estimator. Optimally using auxiliary
information of this type involves positing separate models for the
conditional distributions of $A$ and $M$.

Our proposed methods share connections with the balancing score theory
for causal inference \citep{Rosenbaum&Rubin83}. In particular, note
that the score equations $\Pn D_{A,\tilde g, \hat \lambda}=0$ and
$\Pn D_{M,\tilde g, \hat \lambda}=0$ are balancing equations that
ensure that the empirical mean of $\hat e(W)$ is equal to its
re-weighted mean when using weights $A_i/\tilde g_A(W_i)$ and
$A_iM_i/\tilde g(W_i)$. Covariate balanced estimators have been
traditionally used to reduce bias in observational studies and missing
data models \citep[e.g.,][]{ hainmueller2011entropy,
  imai2014covariate, zubizarreta2015stable}, but covariate selection
for balancing remains an open problem. We conjecture that our theory
may help to solve this problem by shedding light on key transformations
of the covariates that require balance, such as $\hat e(W)$.

We also note that the methods presented may be readily extended to
estimation of other parameters in observational data or randomized
trials. In particular, the estimators for the causal effect of
treatment on the quantile of an outcome presented in
\cite{diaz2015efficient} are amenable to the correction presented
here.

Finally, Donsker \ref{ass:donsker}, which may be restrictive in some
settings, may be removed through the use of a cross-validated version of our
TMLE. Such development would follow from trivial extensions of the
work of \cite{zheng2011cross}, and would be achieved by constructing a
cross-validated version of the MLE in step 2 of the TMLE algorithm
presented in Section~\ref{sec:tmle}.
\begin{appendices}

\section{Proofs}
% \subsection{Lemma~\ref{lemma:bias}}
% Define $\theta_1=P_0 m_1$, and Note that
% $D_{\eta_1,\theta_1}=D_{\eta_1,\theta_0}-\theta_1 + \theta_0$. Under
% \ref{ass:donsker}, an application of Theorem 5.31 of \cite{vanderVaart98} yields
% \begin{equation*}
%   \tmle-\theta_1= \beta(\tilde \eta) +
%   (\Pn - P_0)D_{\eta_1, \theta_1} + o_P\big(n^{-1/2} + |\beta(\tilde\eta)|\big).
% \end{equation*}
% Substituting $D_{\eta_1,\theta_1}=D_{\eta_1,\theta_0}-\theta_1 +
% \theta_0$ yields
% \begin{equation*}
%   \tmle-\theta_0= \beta(\tilde \eta) +
%   (\Pn - P_0)D_{\eta_1, \theta_0} + o_P\big(n^{-1/2} + |\beta(\tilde\eta)|\big).
% \end{equation*}
% The lemma follows by taking expectations and $\lim_{n\to\infty}$ on
% both sides of the above display.
\subsection{Theorem~\ref{lemma:betarep}}
The drift term $\beta(\hat\eta)$ may be decomposed as
\begin{align}
  \beta(\hat\eta)=&\int \frac{1}{\hat g}\{g_0-\hat
                     g\}\{m_0-m_1\}dP_0 +\label{b1}\\
                    &\int \frac{1}{g_1}\{g_0- g_1\}\{m_0-\hat m\}dP_0+\label{b2}\\
                    &\int \frac{1}{\hat g}\{g_1-\hat
                      g\}\{m_1-\hat m\}dP_0 +\label{t1}\\
                    &\int \left\{\frac{1}{\hat
                      g}-\frac{1}{g_1}\right\}\{g_0- g_1\}\{m_1-\hat m\}dP_0+\label{t2}\\
                    &\int \frac{1}{g_1}\{g_0- g_1\}\{m_1-m_0\}dP_0\label{t3}
\end{align}
Under \ref{ass:DR2} we have $(\ref{t1})+(\ref{t2})=o_P(n^{-1/2})$, and
$(\ref{t3})=0$. Denote (\ref{b1}) and (\ref{b2}) with $\beta_g(\hat
g)$ and $\beta_m(\hat m)$, respectively. Then
\begin{equation}\beta(\hat\eta) =  \beta_g(\hat g) +
  \beta_m(\hat m) + o_P(n^{-1/2}),\label{eq:biasdec}
\end{equation}
% where
% \begin{align*}
%         \beta_g(\hat g) &=  \int \frac{1}{\hat g}\{g_0-\hat g\}\{ m_0- m_1\}dP_0\\
%         \beta_m(\hat m)&= \int \frac{1}{\hat g}\{g_0-g_1\}\{ m_0-\hat m\}dP_0.
%                                \end{align*}
Define
\begin{align*}
  \hat\gamma_{A,0}(W) &= P_0\big\{A = 1 \mid m_1(W), \hat m(W)\big\},\\
  \hat\gamma_{M,0}(W) &= P_0\big\{M = 1\mid A = 1, m_1(W), \hat m(W)\big\},\\
  \hat r_{A,0}(W) &= E_{P_0}\left\{\frac{A - g_{A,1}(W)}{g_{A,1}(W)}\mathrel{\bigg|}
                      m_1(W), \hat m(W)\right\},\\
  \hat r_{M,0}(W) &= E_{P_0}\left\{\frac{M - g_{M,1}(W)}{g_1(W)}\mathrel{\bigg|} A=1, m_1(W), \hat m(W)\right\},\\
  \hat e_0(W) &= E_{P_0}\big\{Y - m_1(W)\mid A=1,M=1,g_1(W),\hat g(W)\big\}.
\end{align*}
First, assume $g_1=g_0$, so that
$\beta(\hat\eta) = \beta_g(\hat g) + o_P(n^{-1/2})$. We have
\begin{align}
  \beta_g(\hat g) =& \int \frac{1}{\hat g(w)}\{g_0(w)-\hat
                       g(w)\}\{ m_0(w)- m_1(w)\}dP_0(w)\notag\\
  =& \int \frac{a\,m}{\hat g(w)g_0(w)}\{g_0(w)-\hat
     g(w)\}\{y- m_1(w)\}dP_0(y,m,a,w)\notag\\
  =& \int\left[\int \frac{a\,m}{\hat g(w)g_0(w)}\{y-m_1(w)\}\{g_0(w)-\hat
     g(w)\}dP_0(y \mid a,m,w,g_0(w),\hat g(w))\right]\,dP_0(m,a,w)\notag\\
  =& \int \frac{a\,m\,\hat e_0(w)}{\hat g(w)g_0(w)}\{g_0(w)-\hat
     g(w)\}dP_0(m,a,w)\notag\\
  =& \int \frac{\hat e_0(w)}{\hat g(w)}\{g_0(w)-\hat
     g(w)\}dP_0(w)\notag\\
  =& \int \frac{\hat e_0(w)}{\hat g(w)}\{a\,m-\hat
     g(w)\}dP_0(m,a, w)\notag\\
  =& \int\left[\frac{a\,\hat e_0(w)}{\hat g(w)}\{m - \hat
     g_M(w)\} + \frac{\hat e_0(w)}{\hat g_A(w)}\{a-\hat
     g_A(w)\}\right]dP_0(m,a,w)\notag\\
  =& \int\left[\frac{a\,e_0(w)}{\hat g(w)}\{m - \hat
     g_M(w)\} + \frac{ e_0(w)}{\hat g_A(w)}\{a-\hat
     g_A(w)\}\right]dP_0(m,a,w)\label{DAMs}\\
                     &+ \int\left[\frac{a\,\{\hat e_0(w)-e_0(w)\}}{\hat g(w)}\{m - \hat
                       g_M(w)\} + \frac{\hat e_0(w)-e_0(w)}{\hat g_A(w)}\{a-\hat
                       g_A(w)\}\right]dP_0(m,a,w)\label{OPAM}.
\end{align}
Here $P_0(g_0(w),\hat g(w))$ is the distribution of the
transformation $W\to (g_0(W), \hat g(W))$, where $\hat g$ is
fixed. The third equality follows by the law of iterated expectation and
is obtained by first conditioning on the joint distribution of
$(M, A)$ and the transformations $g_0(W)$ and $\hat
g(W)$.

The term (\ref{DAMs}) is $P_0\{D_{M,\hat g,\lambda_0} +
D_{A,\hat g, \lambda_0}\}$, whereas (\ref{OPAM}) is $O_P\left(||\hat g -
  g_0||^2\right)$. Under \ref{ass:DR2} with $g_1=g_0$ the latter term
is $o_P(n^{-1/2})$, so that
\[\beta_g(\hat g) = P_0\{D_{M,\hat g,\lambda_0} +
  D_{A,\hat g, \lambda_0}\} + o_P(n^{-1/2}).\] The result follows because,
under $g_1=g_0$ we have $e_0(w)=0$, and thus $D_{Y, \hat\mu, \lambda_0}=0$.

Now assume $m_1=m_0$, we have $\beta(\hat\eta) =
\beta_m(\hat m) + o_P(n^{-1/2})$. We have
\begin{align*}
  \beta_m(\hat  m) =& \int\frac{1}{g_1(w)}\{g_0(w) -
                        g_1(w)\}\{ m_0(w) - \hat  m(w)\}dP_0(w)\\
                      =& \int\left\{\frac{g_{A,0}}{g_1(w)}\{g_{M,0}(w) -
                       g_{M,1}(w)\} +\frac{1}{g_{A,1}}\{g_{A,0}-g_{A,1}\}\right\}\{ m_0(w) -
                         \hat  m(w)\}dP_0(w)\\
  =& \int\left\{\frac{a}{g_1(w)}\{m -
     g_{M,1}(w)\} +\frac{1}{g_{A,1}}\{a-g_{A,1}\}\right\}\{ m_0(w) - \hat  m(w)\}dP_0(m,a,w)\\
  =& \int\left[a\hat r_{M,0}(w) + \hat
     r_{A,0}(W)\right]\{ m_0(w) - \hat
     m(w)\}dP_0(m,a,w)\\
  =& \int\left[\hat \gamma_A(w)\hat r_{M,0}(w) + \hat
     r_{A,0}(W)\right]\{ m_0(w) - \hat  m(w)\}dP_0(m,a,w)\\
  =& \int\frac{a\,m\,}{\hat\gamma_{A,0}(w)\hat\gamma_{M,0}(w)}\left[\hat \gamma_A(w)\hat r_{M,0}(w) + \hat
     r_{A,0}(w)\right]\{y - \hat  m(w)\}dP_0(m,a,w)\\
     =& \int a\,m\left[\frac{r_{M,0}(w)}{\gamma_{M,0}(w)} +
        \frac{r_{A,0}(w)}{\gamma_0(w)}\right]\{y - \hat  m(w)\}dP_0(m,a,w) +
        O_P(||\hat m - m_0||^2)
\end{align*}
Under \ref{ass:DR2} with $m_1=m_0$ we have $||\tilde m -
m_0||^2=o_P(n^{-1/2})$ and $r_{A,0}(w)=r_{M,0}(w)=0$. Thus $D_{M,\tilde g, \lambda_0} =
D_{A,\tilde g, \lambda_0}=0$. This completes the proof of the theorem.

\subsection{Theorem~\ref{theo:dr}}
Arguing as in equation (\ref{eq:wh}) we get
\begin{equation*}
  \dtmle-\theta_0  = \beta(\tilde \eta) +
  (\Pn - P_0)D_{\eta_1, \theta_0} + o_P\big(n^{-1/2} +
                   |\beta(\tilde\eta)|\big)
\end{equation*}
Note that, by construction (see
Section~\ref{sec:tmle}), $\hat\beta(\tilde\eta)=0$, so that
Lemma~\ref{lemma:asbeta} below gives us the asymptotic expression for
$\beta(\tilde \eta)$. Substituting this expression we get
\[\tmle-\theta_0 = (\Pn - P_0)(D_{\eta_1, \theta_0} -
  D_{M,g_1,\lambda_0} - D_{A,g_1,\lambda_0} - D_{Y,m_1,\lambda_0}) +
o_P\big(n^{-1/2} + O_P(n^{-1/2})\big).\]
The last term is $o_P(n^{-1/2})$. This, together with the central limit theorem
concludes the proof.
\begin{lemma}[Asymptotic Linearity of $\beta(\hat\eta)$]\label{lemma:asbeta}
  Assume \ref{ass:donsker} and \ref{ass:DR2}. Then
  \begin{equation*}
    \beta(\hat\eta) -\hat\beta(\hat\eta)= -(\Pn - P_0)\{D_{M,g_1,\lambda_0}+D_{A,g_1,\lambda_0} + D_{Y,m_1,\lambda_0}\}
    + o_P(n^{-1/2}).
  \end{equation*}
\end{lemma}
\begin{proof}
  From Theorem~\ref{lemma:betarep}, we have
\[\beta(\hat\eta)=P_0\{D_{A,\hat g,\lambda_0} +
  D_{M,\hat m,\lambda_0} + D_{Y,\hat m,\lambda_0}\} + o_P(n^{-1/2})\]
Next, we show that
$P_0D_{Y,\hat m,\lambda_0} -\Pn D_{Y,\hat m,\hat\lambda} = -(\Pn-P_0)D_{Y,m_1,\lambda_0} + o_P(n^{-1/2}).$ The result
for the other terms follow an analogous analysis.

If $g_1(w)=g_0(w)$ we have $r_{A,0}(w)=r_{M,0}(w)=0$, which implies
$D_{Y,\hat m,\lambda_0}(o) = D_{Y,m_1,\lambda_0}(o)=0$, and the
result follows trivially. If $m_1=m_0$, we have
\[P_0D_{Y,\hat m,\lambda_0}-\Pn D_{Y,\hat m,\hat\lambda} = -(\Pn-P_0)D_{Y,\hat m,\hat\lambda} +
  P_0(D_{Y,\hat m,\lambda_0} - D_{Y,\hat m,\hat\lambda}),\]
where we added and subtracted $P_0D_{Y,\hat m,\hat\lambda}$. We have
\[  P_0(D_{Y,\hat m,\lambda_0} - D_{Y,\hat m,\hat\lambda})
  =\int g_0\left\{\frac{r_{M,0}}{\gamma_{M,0}} - \frac{\hat
      r_M}{\hat\gamma_M} + \frac{r_{A,0}}{\gamma_0} - \frac{\hat
      r_A}{\hat\gamma}\right\}\{m_0 - \hat m\}dP_0
\]
Using the Cauchy-Schwartz and triangle inequalities, we obtain
\[  P_0(D_{Y,\hat m,\lambda_0} - D_{Y,\hat m,\hat\lambda})=O_P\big(||\hat  m -
  m_0||\{||\hat r_A - r_{A,0}|| + ||\hat r_M - r_{M,0}|| +
  ||\hat\gamma_A - \gamma_{A,0}|| + ||\hat\gamma_M - \gamma_{M,0}||\}\big)\]
In light of Lemma~\ref{lemma:rs} below we get
\[  P_0(D_{Y,\hat m,r_0} - D_{Y,\hat m,\hat r})=O_P\big(||\hat  m -
 m_0||\{||\hat g - g_1|| + ||\hat  m -
 m_0|| + n^{-7/20}\}\big).\]
By \ref{ass:DR2} this term is $o_P(n^{-1/2})$.

Under \ref{ass:donsker} and \ref{ass:DR2}, $D_{Y,\hat m,\hat\lambda}$
an application of Theorem 4 of \cite{gine2008uniform} and example 2.10.10
of \cite{vanderVaart&Wellner96} yields that $D_{Y,\hat
  m,\hat\lambda}$ is in a Donsker class. Thus, according to theorem
19.24 of \cite{vanderVaart98}:
$P_0D_{Y,\hat\eta,\lambda_0}-\Pn D_{Y,\hat m,\hat\lambda} = -(\Pn-P_0)D_{Y,\eta_1,\gamma_0} +
o_P(n^{-1/2})$.
\end{proof}
\begin{lemma}\label{lemma:rs} Assume $\hat\gamma_A$, $\hat \gamma_M$, and $\hat\mu$ use the
 bandwidth $\hat h = n^{-0.1}\hopt$ and $K_h$ is a second order kernel. Then
\begin{align*}
||\hat \gamma_A  - \gamma_{A,0}|| &= O_P\big(||\hat g - g_1|| + ||\hat  m -
 m_1|| + n^{-7/20}\big)\\
||\hat \gamma_M  - \gamma_{M,0}|| &= O_P\big(||\hat g - g_1|| + ||\hat  m -
 m_1|| + n^{-7/20}\big)\\
||\hat r_A  - r_{A,0}|| &= O_P\big(||\hat g - g_1|| + ||\hat  m -
                    m_1|| + n^{-7/20}\big)\\
||\hat r_M  - r_{M,0}|| &= O_P\big(||\hat g - g_1|| + ||\hat  m -
                            m_1|| + n^{-7/20}\big)\\
||\hat e  - e_0|| &= O_P\big(||\hat g - g_1|| + ||\hat  m -
 m_1|| + n^{-7/20}\big)
\end{align*}
\end{lemma}
\begin{proof}
We prove the result for $\hat e$. The proofs for the other components
of $\hat \lambda$ follow symmetric arguments. Let
\[\hat e_0(w) = \frac{\sum_{i = 1}^n A_iM_iK_{\hat h}\{g_1(W_i) -
  g_1(w)\}\{Y_i - m_1(W_i)\}}{\sum_{i=1}^nA_iM_iK_{\hat h}\{g_1(W_i) -
  g_1(w)\}}\]
denote the kernel regression estimator that would be computed if
$m_1$ and $g_1$ were known. The triangle inequality yields
\[||\hat e - e_0||\leq ||\hat e - \hat e_0||+||\hat e_0 - e_0||\]
Under the conditions of the lemma, since
$\hat h=n^{-0.1}\hopt$ is an undersmoothing bandwidth,
the leading term of $||\hat e_0 - e_0||^2$ is the variance of a kernel
estimator, which is of order $n^{-1}\hat h^{-1}=O_P(n^{-7/10})$,
which yields $||\hat e_0 - e_0||=O_P(n^{-7/20})$. The first term concerns
estimation of $\mu_1$ and $g_1$ and may be analyzed as follows. To
simplify notation, for a given $g$, let
\[K^\star_{g, i}(x) = \frac{K_{\hat h}\{g(X_i) -
    g(x)\}}{\sum_{i=1}^nK_{\hat h}\{g(X_i) - g(x)\}}.\]
Thus
\begin{align*}
  \hat e(x) - \hat e_0(x) = &\sum_{i=1}^nA_iM_iK^\star_{\hat g,
                                    i}(x)\{Y_i-\hat m(X_i)\} - \sum_{i=1}^nA_iM_iK^\star_{g_1,
                                    i}(x)\{Y_i-m_1(X_i)\} \\
  = &\sum_{i=1}^nA_iM_i\{K^\star_{\hat g, i}(x) - K^\star_{g_1,
      i}(x)\}\{Y_i-\hat m(X_i)\} \\
&+ \sum_{i=1}^nA_iM_iK^\star_{g_1,
      i}(x)\{m_1(X_i)-\hat m(X_i)\}.
\end{align*}
Taking $||\cdot||$ on both sides along with the triangle inequality
yields  the result in
the lemma.
\end{proof}

\section{R code}\label{sec:code}

\end{appendices}
\bibliographystyle{plainnat}
\bibliography{tmle}

\begin{thebibliography}{44}
\providecommand{\natexlab}[1]{#1}
\providecommand{\url}[1]{\texttt{#1}}
\expandafter\ifx\csname urlstyle\endcsname\relax
  \providecommand{\doi}[1]{doi: #1}\else
  \providecommand{\doi}{doi: \begingroup \urlstyle{rm}\Url}\fi

\bibitem[Bang and Robins(2005)]{Bang05}
Heejung Bang and James~M Robins.
\newblock Doubly robust estimation in missing data and causal inference models.
\newblock \emph{Biometrics}, 61\penalty0 (4):\penalty0 962--973, 2005.

\bibitem[Bembom et~al.(2008)Bembom, Fessel, Shafer, and van~der
  Laan]{Bembometal08a}
O.~Bembom, J.W. Fessel, R.W. Shafer, and M.J. van~der Laan.
\newblock Data-adaptive selection of the adjustment set in variable importance
  estimation.
\newblock 2008.
\newblock URL \url{http://www.bepress.com/ucbbiostat/paper231}.

\bibitem[Benkeser et~al.(2016)Benkeser, Carone, van~der Laan, and
  Gilbert]{benkeser2016doubly}
David Benkeser, Marco Carone, Mark~J van~der Laan, and Peter Gilbert.
\newblock Doubly-robust nonparametric inference on the average treatment
  effect.
\newblock Technical Report 356, U.C. Berkeley Division of Biostatistics Working
  Paper Series, 2016.

\bibitem[Bolthausen et~al.(2002)Bolthausen, Perkins, and
  Aad]{bolthausen2002lectures}
Erwin Bolthausen, Edwin Perkins, and van der~Vaart Aad.
\newblock \emph{Lectures on Probability Theory and Statistics: Ecole D'Et{\'e}
  de Probabilit{\'e}s de Saint-Flour XXIX-1999}.
\newblock Springer Science \& Business Media, 2002.

\bibitem[Colantuoni and Rosenblum(2015)]{Colantuoni2015}
Elizabeth Colantuoni and Michael Rosenblum.
\newblock Leveraging prognostic baseline variables to gain precision in
  randomized trials.
\newblock \emph{Statistics in Medicine}, 34\penalty0 (18):\penalty0 2602--2617,
  2015.
\newblock ISSN 1097-0258.
\newblock \doi{10.1002/sim.6507}.
\newblock URL \url{http://dx.doi.org/10.1002/sim.6507}.

\bibitem[D{\'\i}az(2015)]{diaz2015efficient}
Iv{\'a}n D{\'\i}az.
\newblock Efficient estimation of quantiles in missing data models.
\newblock \emph{arXiv preprint arXiv:1512.08110}, 2015.

\bibitem[D{\'i}az et~al.(2016)D{\'i}az, Colantuoni, and Rosenblum]{Diaz2016}
Iv{\'a}n D{\'i}az, Elizabeth Colantuoni, and Michael Rosenblum.
\newblock Enhanced precision in the analysis of randomized trials with ordinal
  outcomes.
\newblock \emph{Biometrics}, 72\penalty0 (2):\penalty0 422--431, 2016.
\newblock ISSN 1541-0420.
\newblock \doi{10.1111/biom.12450}.
\newblock URL \url{http://dx.doi.org/10.1111/biom.12450}.

\bibitem[Gin{\'e} and Nickl(2008)]{gine2008uniform}
Evarist Gin{\'e} and Richard Nickl.
\newblock Uniform central limit theorems for kernel density estimators.
\newblock \emph{Probability Theory and Related Fields}, 141\penalty0
  (3-4):\penalty0 333--387, 2008.

\bibitem[Gruber and van~der Laan(2010)]{Gruber2010t}
Susan Gruber and Mark~J van~der Laan.
\newblock A targeted maximum likelihood estimator of a causal effect on a
  bounded continuous outcome.
\newblock \emph{The International Journal of Biostatistics}, 6\penalty0 (1),
  2010.

\bibitem[Hahn(1998)]{hahn1998role}
Jinyong Hahn.
\newblock On the role of the propensity score in efficient semiparametric
  estimation of average treatment effects.
\newblock \emph{Econometrica}, pages 315--331, 1998.

\bibitem[Hainmueller(2011)]{hainmueller2011entropy}
Jens Hainmueller.
\newblock Entropy balancing for causal effects: A multivariate reweighting
  method to produce balanced samples in observational studies.
\newblock \emph{Political Analysis}, page mpr025, 2011.

\bibitem[Hammer et~al.(1996)Hammer, Katzenstein, Hughes, Gundacker, Schooley,
  Haubrich, Henry, Lederman, Phair, Niu, et~al.]{hammer1996trial}
Scott~M Hammer, David~A Katzenstein, Michael~D Hughes, Holly Gundacker,
  Robert~T Schooley, Richard~H Haubrich, W~Keith Henry, Michael~M Lederman,
  John~P Phair, Manette Niu, et~al.
\newblock A trial comparing nucleoside monotherapy with combination therapy in
  hiv-infected adults with cd4 cell counts from 200 to 500 per cubic
  millimeter.
\newblock \emph{New England Journal of Medicine}, 335\penalty0 (15):\penalty0
  1081--1090, 1996.

\bibitem[Imai and Ratkovic(2014)]{imai2014covariate}
Kosuke Imai and Marc Ratkovic.
\newblock Covariate balancing propensity score.
\newblock \emph{Journal of the Royal Statistical Society: Series B (Statistical
  Methodology)}, 76\penalty0 (1):\penalty0 243--263, 2014.

\bibitem[Juraska et~al.(2012)Juraska, with contributions~from Peter B.~Gilbert,
  Lu, Zhang, Davidian, and Tsiatis]{speff2trial}
Michal Juraska, with contributions~from Peter B.~Gilbert, Xiaomin Lu, Min
  Zhang, Marie Davidian, and Anastasios~A. Tsiatis.
\newblock \emph{speff2trial: Semiparametric efficient estimation for a
  two-sample treatment effect}, 2012.
\newblock URL \url{https://CRAN.R-project.org/package=speff2trial}.
\newblock R package version 1.0.4.

\bibitem[Kang and Schafer(2007)]{Kang2007}
J.~Kang and J.~Schafer.
\newblock Demystifying double robustness: A comparison of alternative
  strategies for estimating a population mean from incomplete data (with
  discussion).
\newblock \emph{Statistical Science}, 22:\penalty0 523--39, 2007.

\bibitem[Koch et~al.(1998)Koch, Tangen, Jung, and Amara]{koch1998issues}
Gary~G Koch, Catherine~M Tangen, Jin-Whan Jung, and Ingrid~A Amara.
\newblock Issues for covariance analysis of dichotomous and ordered categorical
  data from randomized clinical trials and non-parametric strategies for
  addressing them.
\newblock \emph{Statistics in medicine}, 17\penalty0 (15-16):\penalty0
  1863--1892, 1998.

\bibitem[Lee et~al.(2010)Lee, Lessler, and Stuart]{lee2010improving}
Brian~K Lee, Justin Lessler, and Elizabeth~A Stuart.
\newblock Improving propensity score weighting using machine learning.
\newblock \emph{Statistics in medicine}, 29\penalty0 (3):\penalty0 337--346,
  2010.

\bibitem[Moore and van~der Laan(2009)]{moore2009covariate}
Kelly~L Moore and Mark~J van~der Laan.
\newblock Covariate adjustment in randomized trials with binary outcomes:
  Targeted maximum likelihood estimation.
\newblock \emph{Statistics in Medicine}, 28\penalty0 (1):\penalty0 39--64,
  2009.

\bibitem[Neugebauer et~al.(2016)Neugebauer, Schmittdiel, and van~der
  Laan]{neugebauer2016case}
Romain Neugebauer, Julie~A Schmittdiel, and Mark~J van~der Laan.
\newblock A case study of the impact of data-adaptive versus model-based
  estimation of the propensity scores on causal inferences from three inverse
  probability weighting estimators.
\newblock \emph{The international journal of biostatistics}, 12\penalty0
  (1):\penalty0 131--155, 2016.

\bibitem[Polley et~al.(2016)Polley, LeDell, and {van der Laan}]{SL}
Eric Polley, Erin LeDell, and Mark {van der Laan}.
\newblock \emph{SuperLearner: Super Learner Prediction}, 2016.
\newblock URL \url{https://CRAN.R-project.org/package=SuperLearner}.
\newblock R package version 2.0-19.

\bibitem[Porter et~al.(2011)Porter, Gruber, van~der Laan, and
  Sekhon]{Porter2011}
Kristin~E. Porter, Susan Gruber, Mark~J. van~der Laan, and Jasjeet~S. Sekhon.
\newblock The relative performance of targeted maximum likelihood estimators.
\newblock \emph{The International Journal of Biostatistics}, 7\penalty0
  (1):\penalty0 1--34, 2011.

\bibitem[Ridgeway and McCaffrey(2007)]{ridgeway2007}
Greg Ridgeway and Daniel~F. McCaffrey.
\newblock Comment: Demystifying double robustness: A comparison of alternative
  strategies for estimating a population mean from incomplete data.
\newblock \emph{Statist. Sci.}, 22\penalty0 (4):\penalty0 540--543, 11 2007.
\newblock \doi{10.1214/07-STS227C}.
\newblock URL \url{http://dx.doi.org/10.1214/07-STS227C}.

\bibitem[Robins et~al.(2007)Robins, Sued, Lei-Gomez, and Rotnitzky]{Robins2007}
James Robins, Mariela Sued, Quanhong Lei-Gomez, and Andrea Rotnitzky.
\newblock Comment: Performance of double-robust estimators when" inverse
  probability" weights are highly variable.
\newblock \emph{Statistical Science}, 22\penalty0 (4):\penalty0 544--559, 2007.

\bibitem[Robins et~al.(1994)Robins, Rotnitzky, and
  Zhao]{Robins&Rotnitzky&Zhao94}
J.M. Robins, A.~Rotnitzky, and L.P. Zhao.
\newblock Estimation of regression coefficients when some regressors are not
  always observed.
\newblock \emph{Journal of the American Statistical Association}, 89\penalty0
  (427):\penalty0 846--866, September 1994.

\bibitem[Rubin(1987)]{rubin1987multiple}
Donald~B Rubin.
\newblock \emph{Multiple Imputation for Nonresponse in Surveys}.
\newblock John Wiley \& Sons, 1987.

\bibitem[Rubin(1983)]{Rosenbaum&Rubin83}
P.R. Rosenbaum \&~D.B. Rubin.
\newblock The central role of the propensity score in observational studies for
  causal effects.
\newblock \emph{Biometrika}, 70:\penalty0 41--55, 1983.

\bibitem[Tan(2010)]{tan2010bounded}
Zhiqiang Tan.
\newblock Bounded, efficient and doubly robust estimation with inverse
  weighting.
\newblock \emph{Biometrika}, 97\penalty0 (3):\penalty0 661--682, 2010.

\bibitem[van~der Laan(2014)]{van2014targeted}
Mark~J van~der Laan.
\newblock Targeted estimation of nuisance parameters to obtain valid
  statistical inference.
\newblock \emph{The international journal of biostatistics}, 10\penalty0
  (1):\penalty0 29--57, 2014.

\bibitem[van~der Laan(2015)]{van2015generally}
Mark~J van~der Laan.
\newblock A generally efficient targeted minimum loss based estimator.
\newblock Technical Report 343, U.C. Berkeley Division of Biostatistics Working
  Paper Series, 2015.

\bibitem[van~der Laan and Starmans(2014)]{van2014entering}
Mark~J van~der Laan and Richard~JCM Starmans.
\newblock Entering the era of data science: Targeted learning and the
  integration of statistics and computational data analysis.
\newblock \emph{Advances in Statistics}, 2014, 2014.

\bibitem[van~der Laan and Robins(2003)]{vanderLaan2003}
M.J. van~der Laan and J.M. Robins.
\newblock \emph{Unified Methods for Censored Longitudinal Data and Causality}.
\newblock Springer, New York, 2003.

\bibitem[van~der Laan and Rose(2011)]{vanderLaanRose11}
M.J. van~der Laan and S.~Rose.
\newblock \emph{Targeted Learning: Causal Inference for Observational and
  Experimental Data}.
\newblock Springer, New York, 2011.

\bibitem[van~der Laan and Rubin(2006)]{vanderLaan&Rubin06}
M.J. van~der Laan and D.~Rubin.
\newblock Targeted maximum likelihood learning.
\newblock \emph{The International Journal of Biostatistics}, 2\penalty0
  (1):\penalty0 Article 11, 2006.

\bibitem[van~der Laan et~al.(2005)van~der Laan, Petersen, and
  Joffe]{vanderLaan&Petersen&Joffe05}
M.J. van~der Laan, M.L. Petersen, and M.M. Joffe.
\newblock History-adjusted marginal structural models \& statically-optimal
  dynamic treatment regimens.
\newblock \emph{The International Journal of Biostatistics}, 1\penalty0
  (1):\penalty0 10--20, 2005.

\bibitem[van~der Laan et~al.(2007)van~der Laan, Polley, and
  Hubbard]{vanderLaan&Polley&Hubbard07}
M.J. van~der Laan, E.~Polley, and A.~Hubbard.
\newblock Super learner.
\newblock \emph{Statistical Applications in Genetics \& Molecular Biology},
  6\penalty0 (25):\penalty0 Article 25, 2007.

\bibitem[van~der Laan(2006)]{Wang&Bembom&vanderLaan06}
Y.~Wang \& O. Bembom \&~M.J. van~der Laan.
\newblock Data adaptive estimation of the treatment specific mean.
\newblock \emph{Journal of Statistical Planning \& Inference}, 2006.

\bibitem[van~der Vaart(1998)]{vanderVaart98}
A.~W. van~der Vaart.
\newblock \emph{Asymptotic Statistics}.
\newblock Cambridge University Press, 1998.

\bibitem[van~der Vaart and Wellner(1996)]{vanderVaart&Wellner96}
A.~W. van~der Vaart and J.~A. Wellner.
\newblock \emph{Weak {C}onvergence and {E}mprical {P}rocesses}.
\newblock Springer-Verlag New York, 1996.

\bibitem[van~der Vaart et~al.(2006)van~der Vaart, Dudoit, and van~der
  Laan]{vanderVaart&Dudoit&vanderLaan06}
A.W. van~der Vaart, S.~Dudoit, and M.J. van~der Laan.
\newblock Oracle inequalities for multi-fold cross-validation.
\newblock \emph{Statistics \& Decisions}, 24\penalty0 (3):\penalty0 351--371,
  2006.

\bibitem[Vermeulen and Vansteelandt(2015)]{vermeulen2015bias}
Karel Vermeulen and Stijn Vansteelandt.
\newblock Bias-reduced doubly robust estimation.
\newblock \emph{Journal of the American Statistical Association}, 110\penalty0
  (511):\penalty0 1024--1036, 2015.

\bibitem[Vermeulen and Vansteelandt(2016)]{vermeulen2016data}
Karel Vermeulen and Stijn Vansteelandt.
\newblock Data-adaptive bias-reduced doubly robust estimation.
\newblock \emph{The international journal of biostatistics}, 12\penalty0
  (1):\penalty0 253--282, 2016.

\bibitem[Zhang et~al.(2008)Zhang, Tsiatis, and Davidian]{Zhang2008}
Min Zhang, Anastasios~A Tsiatis, and Marie Davidian.
\newblock Improving efficiency of inferences in randomized clinical trials
  using auxiliary covariates.
\newblock \emph{Biometrics}, 64\penalty0 (3):\penalty0 707--715, 2008.

\bibitem[Zheng and van~der Laan(2011)]{zheng2011cross}
Wenjing Zheng and Mark~J van~der Laan.
\newblock Cross-validated targeted minimum-loss-based estimation.
\newblock In \emph{Targeted Learning}, pages 459--474. Springer, 2011.

\bibitem[Zubizarreta(2015)]{zubizarreta2015stable}
Jos{\'e}~R Zubizarreta.
\newblock Stable weights that balance covariates for estimation with incomplete
  outcome data.
\newblock \emph{Journal of the American Statistical Association}, 110\penalty0
  (511):\penalty0 910--922, 2015.

\end{thebibliography}
\end{document}